\documentclass{amsart}
\usepackage{txfonts}
\usepackage{amsfonts}
\usepackage{amssymb}
\usepackage{mathrsfs}
\numberwithin{equation}{section}
 \newtheorem{thm}{Theorem}
 
 \newtheorem{prop}[thm]{Proposition}
 \newtheorem{lem}[thm]{Lemma}
 \newtheorem{cor}[thm]{Corollary}
 \newtheorem{defn}{Definition}
 \newtheorem{rem}{Remark}
  \newtheorem{exm}{Example}
  \newcommand{\f}{\mathbb{F}_{q}}
  \newcommand{\fl}{\mathbb{F}_{5}}
   \newcommand{\p}{\mathcal {P}}
\begin{document}
\title{Multi-receiver Authentication Scheme for Multiple Messages Based on Linear Codes}
\author{Jun Zhang, Xinran Li and Fang-Wei Fu}
\address{Chern Institute of Mathematics, Nankai University, Tianjin, P.R. China}
\email{zhangjun04@mail.nankai.edu.cn; xinranli@mail.nankai.edu.cn; fwfu@nankai.edu.cn}
\thanks{This research is supported by the National Key Basic Research Program of China (Grant No. 2013CB834204), and the National Natural Science Foundation of
China (Nos. 61171082, 10990011 and 60872025). The author Jun Zhang is also supproted by the Chinese Scholarship Council under the State Scholarship Fund during visiting University of California, Irvine.}
\date{}

 \maketitle
\begin{abstract}
In this paper, we construct an authentication scheme for multi-receivers and multiple messages based on a linear code $C$. This construction can be regarded as a generalization of the authentication scheme given by Safavi-Naini and Wang \cite{SW}. Actually, we notice that the scheme of Safavi-Naini and Wang is constructed with Reed-Solomon codes. The generalization to linear codes has the similar advantages as generalizing Shamir's secret sharing scheme to linear secret sharing sceme based on linear codes \cite{CC,M2,M1,MS,Shamir}. For a fixed message base field $\f$, our scheme allows arbitrarily many receivers to check the integrity of their own messages, while the scheme of Safavi-Naini and Wang has a constraint on the number of verifying receivers $V\leqslant q$. And we introduce access structure in our scheme. Massey~\cite{M1} characterized the access structure of linear secret sharing scheme by minimal codewords in the dual code whose first component is $1$. We slightly modify the definition of minimal codewords in~\cite{M1}. Let $C$ be a $[V,k]$ linear code. For any coordinate $i\in \{1,2,\cdots,V\}$, a codeword $\vec{c}$ in $C$ is called \emph{minimal respect to $i$} if the codeword $\vec{c}$ has component $1$ at the $i$-th coordinate and there is no other codeword whose $i$-th component is $1$ with support strictly contained in that of $\vec{c}$. Then the security of receiver $R_i$ in our authentication scheme is characterized by the minimal codewords respect to $i$ in the dual code $C^\bot$.

\end{abstract}

\section{Introduction}
\subsection{Background}
One of the important goals of cryptographic scheme is authentication, which is concerned with the approaches of providing data integrity and data origin validation between two communication entities in computer network. Traditionally, it simply deals with the data authentication problem from a single sender to a single receiver.
With the rapid progress of network communication, the urgent need for providing data authentication has escalated to multi-receiver and/or multi-sender scenarios. However, the original point-to-point authentication techniques are not suitable for multi-point communication. In the multi-receiver authentication model, a sender broadcasts an authenticated message such that all the receivers can independently verify the authenticity of the message with their own private keys. It requires a security that malicious groups of up to a given size of receivers can not successfully impersonate the transmitter, or substitute a transmitted message. Desmedt et al. \cite{DFY} gave an authentication scheme of single message for multi-receivers. Safavi-Naini and Wang \cite{SW} extended the DFY scheme \cite{DFY} to be an authentication scheme of multiple messages for multi-receivers.

The receivers independently verify the authenticity of the message using each own private key. So multi-receiver authentication scheme involves a procedure of secret sharing. To introduce the linear secret sharing scheme based on linear codes, we recall some definitions in coding theory.

Let $\f^V$ be the $V$-dimensional vector space over the finite field $\f$ with $q$ elements. For any vector $\vec{x}=(x_1,x_2,\cdots,x_V)\in \f^V$, the \emph{Hamming weight} $\mathrm{Wt}(\vec{x})$ of $\vec{x}$ is defined to be the number of non-zero coordinates, i.e.,
$$\mathrm{Wt}(\vec{x})=\#\left\{i\,|\,1\leqslant i\leqslant V,\,x_i\neq 0\right\}\ .$$
A \emph{linear $[V,k]$ code} $C$ is a $k$-dimensional linear subspace of $\f^V$. The \emph{minimum distance} $d(C)$ of $C$ is the minimum Hamming weight of all non-zero vectors in $C$, i.e.,
$$d(C)=\min\{\mathrm{Wt}(\vec{c})\,|\,\vec{c}\in C\setminus\{\vec{0}\}\}\ .$$
A linear $[V,k]$ code $C\subseteq \f^V$ is called a \emph{$[V,k,d]$ linear code} if $C$ has minimum distance $d$. A vector in $C$ is called a $codeword$ of $C$. A matrix $G\in \f^{k\times V}$ is call a \emph{generator matrix} of $C$ if rows of $G$ form a basis for $C$. A well known trade-off between the parameters of a linear $[V,k,d]$ code is the Singleton bound which states that
$$d\leqslant V-k+1\ .$$
 A $[V,k,d]$ code is called a \emph{maximum distance separable} (MDS) code if $d=V-k+1$. The \emph{dual code} $C^\bot$ of $C$ is defined as the set
\[
  \left\{\vec{x}\in \f^V\,|\,\vec{x}\cdot\vec{c}=0\,\mbox{for all }\vec{c}\in C\right\},
\]
where $\vec{x}\cdot\vec{c}$ is the inner product of vectors $\vec{x}$ and $\vec{c}$, i.e.,
\[
    \vec{x}\cdot\vec{c}=x_1c_1+x_2c_2+\cdots+x_Vc_V\ .
\]

The secret sharing scheme provides security of a secret key by ``splitting" it to several parts which are kept by different persons. In this way, it might need many persons to recover the original key. It can achieve to resist the attack of malicious groups of persons. Shamir \cite{Shamir} used polynomials over finite fields to give an $(S,T)$ threshold secret sharing scheme such that any $T$ persons of the $S$ shares can uniquely determine the secret key but any $T-1$ persons can not get any information of the key. A linear secret sharing scheme based on a linear code~\cite{M1} is constructed as follows: encrypt the secret to be the first coordinate of a codeword and distribute the rest of the codeword (except the first secret coordinate) to the group of shares. McEliece and Sarwate \cite{MS} pointed out that the Shamir's construction is essentially a linear secret sharing scheme based on Reed-Solomon codes. Also as a natural generalization of Shamir'construction, Chen and Cramer \cite{CC} constructed a linear secret sharing scheme based on algebraic geometric codes.

The \emph{qualified subset} of a linear secret sharing scheme is a subset of shares such that the shares in the subset can recover the secret key. A qualified subset is call \emph{minimal} if any share is removed from the qualified subset, the rests cannot recover the secret key. The \emph{access structure} of a linear secret sharing scheme consists of all the minimal qualified subsets. A codeword $\vec{v}$ in a linear code $C$ is said to be
\emph{minimal} if $\vec{v}$ is a non-zero codeword whose leftmost
nonzero component is a $1$ and no other codeword
$\vec{v}^{\prime}$ whose leftmost nonzero component is $1$ has support strictly contained in the support of $\vec{v}$. Massey \cite{M1,M2} showed that the access structure of a linear secret sharing scheme based on a linear code are completely determined by the minimal codewords in the dual code whose first component is $1$.
\begin{prop}[\cite{M1}]\label{minimal}
The access structure of the linear secret-sharing scheme corresponding to
the linear code $C$ is specified by those minimal
codewords in the dual code $C^{\perp}$ whose first component is $1$. In the manner that the set of shares specified by a
minimal codeword whose first component is $1$ in the dual code is the set of shares corresponding
to those locations after the first in the support of this minimal codeword.
\end{prop}

In both schemes of Desmedt et al. \cite{DFY} and Safavi-Naini and Wang~\cite{SW}, the key distribution is similar to that in Shamir's secret sharing scheme \cite{Shamir}, using polynomials. Both schemes are $(V,k)$ threshold authentication scheme, i.e., any malicious groups of up to $k-1$ receivers can not successfully ( unconditional secure in the meaning of information theory) impersonate the transmitter, or substitute a transmitted message to any other receiver, while any $k$ receivers or more receivers can successfully impersonate the transmitter, or substitute a transmitted message to any other receiver. Actually, in the proof of security of the authentication scheme of Safavi-Naini and Wang, the security is equivalent to the difficulty to recover the private key of other receivers. So the security essentially depends on the security of key distribution.

In this paper, we use general linear codes to generalize the scheme of Safavi-Naini and Wang. One advantage is that our scheme allows arbitrarily many verifying receivers for a fixed message base field $\f$, while the scheme of Safavi-Naini and Wang has a constraint on the number of verifying receivers $V\leqslant q$. We introduce the concept of minimal codeword respect to each coordinate, which helps to characterize the capability of resisting substitution attack in our authentication scheme, similarly to the linear secret sharing scheme~\cite{M2}. It guarantees higher security for some important receivers.

\subsection{Our Construction and Main Results}
In a multi-receiver authentication model for multiple messages, a trusted authority choose random parameters as the secret key and generates shares of private keys secretly. Then the trusted authority transmits a private key to each receiver and secret parameters to the source. For each fixed message, the source computes the authentication tag using the secret parameters and sends the message adding with the tag. In the verification phase, the receiver verify the integrity of each tagged message using his private key. There are some malicious receivers who collude to perform an impersonation attack by constructing a fake message, or a substitution attack by altering the message content such that the new tagged message can be accepted by some other receiver or specific receiver.

 In this subsection, we present our construction of an authentication scheme based on a linear code for multi-receivers and multiple messages. It will be shown that the ability of our scheme to resist the attack of the malicious receivers is measured by the minimum distance of the dual code and minimal codewords respect to specific coordinate in the dual code.

Let $C\subseteq \f^V$ be a linear code with minimum distance $d(C)\geqslant 2$. And assume that the minimum distance of the dual code $C^\bot$ is $d(C^\bot)\geqslant 2$. Fix a generator matrix $G$ of $C$

\begin{equation*}
    G=\left(
  \begin{array}{cccc}
    g_{1,1} & g_{1,2} & \cdots & g_{1,V} \\
     g_{2,1} & g_{2,2} & \cdots & g_{2,V} \\
     \vdots & \vdots & \ddots & \vdots \\
    g_{k,1} & g_{k,2} & \cdots & g_{k,V}\\
  \end{array}
\right)\ .
\end{equation*}
Then make $G$ public. Our scheme is as follows.

\begin{itemize}
  \item \textbf{Key generation:} A trusted authority randomly chooses parameters
  \begin{equation*}
    A=\left(
  \begin{array}{cccc}
    a_{0,1} & a_{0,2} & \cdots & a_{0,k} \\
     a_{1,1} & a_{1,2} & \cdots & a_{1,k} \\
     \vdots & \vdots & \ddots & \vdots \\
    a_{M,1} & a_{M,2} & \cdots & a_{M,k}\\
  \end{array}
\right)\in \f^{(M+1)\times k}\ .
\end{equation*}
  \item \textbf{Key distribution:}  The trusted authority computes
    \begin{equation*}
    B=A\cdot G=\left(
  \begin{array}{cccc}
    b_{0,1} & b_{0,2} & \cdots & b_{0,V} \\
     b_{1,1} & b_{1,2} & \cdots & b_{1,V} \\
     \vdots & \vdots & \ddots & \vdots \\
    b_{M,1} & b_{M,2} & \cdots & b_{M,V}\\
  \end{array}
\right)\ .
\end{equation*}
  Then the trusted authority distributes each receiver $R_i$ the $i$-th column of $B$ as his private key, for $i=1,2,\cdots,V$.

  \item \textbf{Authentication tag:} For message $s\in \f$, the source computes the
 tag map
      \begin{equation*}
  \begin{array}{cccc}
    L=[L_1,L_2,\cdots,L_k]: & \f & \rightarrow & \f^{k} \\
     & s &\mapsto &[L_1(s),L_2(s),\cdots,L_k(s)] \ ,
  \end{array}
\end{equation*}
where the map $L_i$ ($i=1,2,\cdots,k$) is defined by
\[
   L_i(s)=\sum_{j=0}^M a_{j,i}s^{j}\ .
\]
Instead of sending the message $s\in \f$, the source actually sends the authenticated messages $\vec{x}$ of the form\footnote{In general, we can first use a hash function $h:\{0,1\}^{*}\rightarrow\f$ to hash the message $s$, then send the tagged message $[s, L(h(s))]$.}
\[
  \vec{x}=[s, L(s)]\in \f^{1+k}\ .
\]
\item \textbf{Verification:}The receiver $R_i$ accepts the message $[s, L(s)]$ if $ \sum_{t=0}^M s^{t} b_{t,i}= \sum_{j=1}^{k} L_j(s)g_{j,i}$. Under the integrity of the tagged message, one can easily verify the following
\begin{equation*}
    \sum_{t=0}^M s^{t} b_{t,i}=\sum_{t=0}^M s^{t}\sum_{j=1}^k a_{t,j}g_{j,i}=\sum_{j=1}^k (\sum_{t=0}^M a_{t,j}s^{t})g_{j,i}= \sum_{j=1}^{k} L_j(s)g_{j,i}\ .
\end{equation*}
Here, we call the result $\sum_{t=0}^M s^{t} b_{t,i}$ the \emph{label} of $R_i$ for message $s$.
\end{itemize}
If we take $C$ to be the Reed-Solomon code, i.e., the generator matrix $G$ is of the form
\begin{equation}\label{rs}
     G=\left(
   \begin{array}{cccc}
     1 & 1 & \cdots & 1 \\
     x_1 & x_2 & \cdots & x_V \\
     x_1^2 & x_2^2 & \cdots & x_V^2 \\
      \vdots & \vdots & \ddots & \vdots \\
     x_1^{k-1} & x_2^{k-1} & \cdots & x_V^{k-1}
   \end{array}
 \right)\ ,
 \end{equation}
  for pairwise distinct $x_1,x_2,\cdots,x_V\in \f$, then the scheme is the scheme of Safavi-Naini and Wang~\cite{SW}.

The security of the above authentication scheme is summarized in the following theorems.
\begin{thm}
The scheme we constructed above is a unconditionally secure multi-receiver authentication code against a coalition of up to ($d(C^\bot)-2$) malicious receivers in which every key can be used to authentication up to $M$ messages.
\end{thm}

More specifically, if we consider what a coalition of malicious receivers can successfully make a substitution attack to one fixed receiver $R_i$. To characterize this malicious group, we slightly modify the definition of minimal codeword in \cite{M1}.
\begin{defn}
Let $C$ be a $[N,k]$ linear code. For any $i\in \{1,2,\cdots,N\}$, a codeword $\vec{c}$ in $C$ is called \emph{minimal respect to $i$} if the codeword $\vec{c}$ has component $1$ at the $i$-th location and there is no other codeword whose $i$-th component is $1$ with support strictly contained in that of $\vec{c}$.
\end{defn}
Then we have
\begin{thm}\label{thm2}For the authentication scheme we constructed, we have
\begin{description}
\item[(i)] The set of all minimal malicious groups that can successfully make a substitution attack to the receiver $R_{i}$ is determined completely by all the minimal codewords respect to $i$ in the dual code $C^\bot$.

\item[(ii)] All malicious groups that can not produce a fake authenticated message which can be accepted by the receiver $R_{i}$ are one-to-one corresponding to subsets of $[V]\setminus\{i\}$ such that each of them together with $i$ does not contain any support of minimal codeword respect to $i$ in the dual code $C^\bot$, where $[V]=\{1,2,\cdots,V\}$.
\end{description}
\end{thm}

Compared with Safavi-Naini and Wang's scheme, our scheme has an important advantage. The scheme of Safavi-Naini and Wang is a $(V,k)$ threshold authentication scheme, so any coalition of $k$ malicious receivers can easily make a substitution attack to any other receiver. While in our scheme, by Theorem \ref{thm2}, sometimes it can withstand the attack of coalitions of $k$ or more malicious receivers to some fixed important receiver $R_i$. And it is in general {\bf NP}-hard to find one (or list all) coalition(s) of malicious receivers with the minimum members that can make a substitution attack to the receiver $R_i$. So in this sense, our scheme has better security than the previous one.

The rest of this paper is organized as follows. In Section~2, we give the security analysis of our scheme. In Section~3, we show the relationship between the security of our scheme and parameters of the linear code.

\section{Security Analysis of Our Authentication Scheme}
In this section, we present the security analysis of our scheme. From the verification step, we notice that a tagged message $[s,v_1,v_2,\cdots,v_k]$ can be accepted by the receiver $R_i$ if and only if $\sum_{t=0}^M s^{t} b_{t,i}= \sum_{j=1}^{k} v_jg_{j,i}$. So in order to make a substitution attack to $R_i$, it suffices to know the label $\sum_{t=0}^M s^{t} b_{t,i}$ for some $s\in \f$ not sent by the transmitter, then it is trivial to construct a tag $(v_1,v_2,\cdots,v_k)$ such that $\sum_{t=0}^M s^{t} b_{t,i}= \sum_{j=1}^{k} v_jg_{j,i}$.

Indeed, we will find that the security of the above authentication scheme depends on the hardness of finding the key matrix $A$ from a system of linear equations. Suppose a group of $K$ malicious receivers collaborate to recover $A$ and make a substitution attack. Without loss of generality, we assume that the malicious receivers are $R_1,R_2,\cdots,R_K$. Suppose $s_1,s_2,\cdots,s_M$ have been sent. Each $R_i$ has some information about the key $A$:
\begin{gather*}\label{equation1}
   \left(
     \begin{array}{ccccc}
      1& s_1 & s_1^2 & \cdots & s_1^{M} \\
      1& s_2 & s_2^2 & \cdots & s_2^{M}\\
      \vdots & \vdots & \vdots & \ddots & \vdots \\
     1& s_M & s_M^2 & \cdots & s_M^{M}\\
     \end{array}
   \right)\cdot A=
   \left(
     \begin{array}{cccc}
       L_1(s_1) & L_2(s_1) & \cdots & L_k(s_1) \\
        L_1(s_2) & L_2(s_2) & \cdots & L_k(s_2)\\
       \vdots & \vdots & \ddots & \vdots \\
       L_1(s_M) & L_2(s_M) & \cdots & L_k(s_M)\\
     \end{array}
   \right)
\end{gather*}
and
\begin{gather*}
    A\cdot \left(
                       \begin{array}{c}
                          g_{1,i} \\
                          g_{2,i} \\
                          \vdots\\
                          g_{k,i} \\
                        \end{array}
                      \right)=\left(
                       \begin{array}{c}
                          b_{0,i} \\
                          b_{1,i} \\
                          \vdots\\
                          b_{M,i} \\
                        \end{array}
                      \right)\ .
\end{gather*}
The group of malicious receivers combines their equations, and they get a system of linear equations
\begin{equation}\label{equation}
\left\{
  \begin{array}{c}
     \left(
     \begin{array}{ccccc}
      1& s_1 & s_1^2 & \cdots & s_1^{M} \\
      1& s_2 & s_2^2 & \cdots & s_2^{M}\\
      \vdots & \vdots & \vdots & \ddots & \vdots \\
     1& s_M & s_M^2 & \cdots & s_M^{M}\\
     \end{array}
   \right)\cdot A=
   \left(
     \begin{array}{cccc}
       L_1(s_1) & L_2(s_1) & \cdots & L_k(s_1) \\
        L_1(s_2) & L_2(s_2) & \cdots & L_k(s_2)\\
       \vdots & \vdots & \ddots & \vdots \\
       L_1(s_M) & L_2(s_M) & \cdots & L_k(s_M)\\
     \end{array}
   \right),
\\
    A\cdot \left(
  \begin{array}{cccc}
    g_{1,1} & g_{1,2} & \cdots & g_{1,K} \\
     g_{2,1} & g_{2,2} & \cdots & g_{2,K} \\
     \vdots & \vdots & \ddots & \vdots \\
    g_{k,1} & g_{k,2} & \cdots & g_{k,K}\\
  \end{array}
\right)=\left(
  \begin{array}{cccc}
    b_{0,1} & b_{0,2} & \cdots & b_{0,K} \\
     b_{1,1} & b_{1,2} & \cdots & b_{1,K} \\
     \vdots & \vdots & \ddots & \vdots \\
    b_{M,1} & b_{M,2} & \cdots & b_{M,K}\\
  \end{array}
\right)\ .
  \end{array}
\right.
\end{equation}

\begin{lem}\label{keylem}
Let $P$ be the subspace of $\f^k$ generated by $\{g_j\,|\,j=1,2,\cdots,K\}$, where $g_j$ represents the $j$-th column of the
generator matrix $G$. Suppose $K_0=\dim{P}\leqslant k-1$. Then there exists exact $q^{k-K_0}$ matrices $A$ satisfying the system of equations (\ref{equation}). \end{lem}
\begin{proof}
Denote
\begin{equation*}
 S_M=   \left(
      \begin{array}{ccccc}
        1 & s_1 & s_1^{2} & \cdots & s_1^{M} \\
        1 & s_2 & s_2^{2} & \cdots & s_2^{M} \\
        \vdots & \vdots & \vdots & \ddots & \vdots \\
         1 & s_M & s_M^{2} & \cdots & s_M^{M} \\
      \end{array}
    \right) \ .
\end{equation*}
Rewrite the matrix $A$ of variables $a_{i,j}$ as a single column of $k(M+1)$ variables. Then System~(\ref{equation}) becomes
 \begin{equation}\label{equ2}
    \left(
      \begin{array}{cccc}
        S_M &  &  &  \\
         & S_M &  &  \\
         &  & \ddots &  \\
         &  &  & S_M \\
        g_{1,1}\vec{I}_{M+1} & g_{2,1}\vec{I}_{M+1} & \cdots & g_{k,1}\vec{I}_{M+1} \\
        g_{1,2}\vec{I}_{M+1} & g_{2,2}\vec{I}_{M+1} & \cdots & g_{k,2}\vec{I}_{M+1} \\
        \vdots & \vdots & \ddots & \vdots \\
        g_{1,K}\vec{I}_{M+1} & g_{2,K}\vec{I}_{M+1} & \cdots & g_{k,K}\vec{I}_{M+1} \\
      \end{array}
    \right)\cdot
    \left(
      \begin{array}{c}
        a_{0,1} \\
        a_{1,1} \\
        \vdots \\
        a_{M,1} \\
        a_{0,2}\\
        a_{1,2}\\
        \vdots \\
        a_{M,2}\\
        \vdots \\
         a_{0,k}\\
         a_{1,k}\\
         \vdots \\
         a_{M,k}\\
      \end{array}
    \right)=T
    \end{equation}
where $\vec{I}_{M+1}$ is the identity matrix with rank ($M+1$) and $T$ is the column vector of constants in System (\ref{equation}) with proper order. Notice that the space generated by rows of $S_M$ is contained in the space $\f^{M+1}$ generated by $g_{i,j}\vec{I}_{M+1}$ if $g_{i,j}\neq 0$.
  So the rank of the big matrix of coefficients in System~(\ref{equ2}) equals to
  \[
      M\cdot k+K_0
  \]
which is less than $k(M+1)$, the number of variables. So System (\ref{equ2}) has $q^{k(M+1)-kM-K_0}=q^{k-K_0}$ solutions, i.e., System (\ref{equation}) has $q^{k-K_0}$ solutions. 
\end{proof}
\begin{rem}
 In \cite{SW}, they gave a constructive proof of Lemma \ref{keylem} in the case that $G$ is of the form~(\ref{rs}).
The method here can be used for a general class of systems of linear equations over a field $F$:
 \begin{equation*}
    \left\{
      \begin{array}{c}
       D\cdot X=C_1 \\
        X\cdot Z=C_2
      \end{array}
    \right.
 \end{equation*}
where $X$ is a $m\times n$ matrix of variables, the coefficient matrices $D\in F^{g\times m}$ with rank $\leqslant m-1$ and $Z\in F^{n\times h}$ with rank $\leqslant n-1$, the constant matrices $C_1\in F^{g\times n}$ and $C_2\in F^{m\times h}$. Then solutions of the system in $F^{m\times n}$ has $(m-g)(n-h)$-dimensional hypersurface in the space $F^{m\times n}$.
\end{rem}

Note that if $C\,[n,k,d=n-k+1]$ is an MDS code, e.g., Reed-Solomon code, then whenever $K\leqslant k-1$ the vectors in any $K$-subset of columns of $G$ are linearly independent.

By Lemma \ref{keylem}, the security of our authentication scheme follows.
\begin{thm}\label{general}
The scheme we constructed above is an unconditionally secure multi-receiver authentication scheme against a coalition of up to ($d(C^\bot)-2$) malicious receivers in which every key can be used to authentication up to $M$ messages.
\end{thm}

\begin{proof} Suppose the source receiver has sent messages $s_1,s_2,\cdots,s_{M}$.
It is enough to consider the case that $K=d(C^\bot)-2$ malicious receivers $R_1,\cdots,R_K$ have received the $M$ messages, since in this case they know the most information about the key matrix $A$.

What they try to do is to guess the label $b_{0,K+1}+b_{1,K+1}s_{M+1}+b_{2,K+1}s_{M+1}^{2}+\cdots+b_{M,K+1}s_{M+1}^{M}$ for some $s_{M+1}\notin \{s_1,s_2,\cdots,s_{M}\}$ and construct a vector $(v_1,v_2,\cdots,v_k)$ such that
\[
    \sum_{i=1}^k v_{i}g_{i,K+1}=b_{0,K+1}+b_{1,K+1}s_{M+1}+b_{2,K+1}s_{M+1}^{2}+\cdots+b_{M,K+1}s_{M+1}^{M}\ .
\]
Then the fake message $[s_{M+1}, v_1,v_2,\cdots,v_k]$ can be accepted by $R_{K+1}$.

Because any $K=d(C^\bot)-2$ columns of the generator matrix $G$ is linearly independent over $\f$, otherwise there exist $x_1,\cdots,x_{K}\in \f$ such that $\sum_{j=1}^Kx_j\vec{g}_{j}=\vec{0}$ where $\vec{g}_{j}$ is the $j$-th column of $G$, then the dual code $C^\bot$ will have a codeword $(x_1,\cdots,x_{K},0,\cdots,0)$ with Hamming weight $\leqslant d(C^\bot)-2$ which is a contradiction. By Lemma \ref{keylem}, there exists $q^{k-d(C^\bot)+2}$ matrices $A$ satisfying the system of equations~(\ref{equation}).

For any $s_{M+1}\notin \{s_1,s_2,\cdots,s_{M}\}$, we define the label map
\begin{equation*}
    \begin{array}{rccc}
      \varphi_{s_{M+1}}: & \{\textrm{Solutions of System (\ref{equation})}\} & \longrightarrow & \f \\
       & A & \mapsto & (1,s_{M+1},s_{M+1}^2,\cdots,s_{M+1}^M)A\left(\begin{array}{c}
                                                                g_{1,K+1} \\
                                                                g_{2,K+1} \\
                                                                \vdots \\
                                                                g_{k,K+1}
                                                              \end{array}\right)\ .
    \end{array}
\end{equation*}
Then we claim:
\begin{description}
\item[(1)] $\varphi_{s_{M+1}}$ is surjective.
\item[(2)] for any $y\in \f$, the number of the inverse image of $y$ is $\#\varphi_{s_{M+1}}^{-1}(y)=q^{k-d(C^\bot)+1}$.
\end{description}
So the information held by the colluders allows them to calculate $q$ equally likely different labels for $s_{M+1}$ and hence their probability of success is $1/q$ which is equal to that of guessing a label $b_{0,K+1}+b_{1,K+1}s_{M+1}+b_{2,K+1}s_{M+1}^{2}+\cdots+b_{M,K+1}s_{M+1}^{M}$ for $s_{M+1}$ randomly from $\f$. And hence we finish the proof of the theorem.

Next, we prove our claim. As $K+1=d(C^\bot)-1$, $g_1,g_2,\cdots,g_{K+1}$ is linearly independent over $\f$, otherwise the dual code $C^\bot$ will have a codeword with Hamming weight $\leqslant d(C^\bot)-1$ which is impossible by the definition of minimum distance of a code. Then choose $k-K-1=k-d(C^\bot)+1$ extra columns of $G$ such that they combining with $g_1,g_2,\cdots,g_{K+1}$ form a basis of $\f^k$. Without loss of generality, we assume the first $k$ columns of $G$ is linearly independent of $\f$. For any $P\in \f^{(M+1)\times(k-d(C^\bot)+1)}$, the system of linear equations
\begin{equation}\label{equ3}
\left\{
  \begin{array}{rl}
     \left(
     \begin{array}{ccccc}
      1& s_1 & s_1^2 & \cdots & s_1^{M} \\
      1& s_2 & s_2^2 & \cdots & s_2^{M}\\
      \vdots & \vdots & \vdots & \ddots & \vdots \\
     1& s_M & s_M^2 & \cdots & s_M^{M}\\
     \end{array}
   \right)\cdot A=&
   \left(
     \begin{array}{cccc}
       L_1(s_1) & L_2(s_1) & \cdots & L_k(s_1) \\
        L_1(s_2) & L_2(s_2) & \cdots & L_k(s_2)\\
       \vdots & \vdots & \ddots & \vdots \\
       L_1(s_M) & L_2(s_M) & \cdots & L_k(s_M)\\
     \end{array}
   \right),
\\
    A\cdot \left(
  \begin{array}{cccc}
    g_{1,1} & g_{1,2} & \cdots & g_{1,K} \\
     g_{2,1} & g_{2,2} & \cdots & g_{2,K} \\
     \vdots & \vdots & \ddots & \vdots \\
    g_{k,1} & g_{k,2} & \cdots & g_{k,K}\\
  \end{array}
\right)=&\left(
  \begin{array}{cccc}
    b_{0,1} & b_{0,2} & \cdots & b_{0,K} \\
     b_{1,1} & b_{1,2} & \cdots & b_{1,K} \\
     \vdots & \vdots & \ddots & \vdots \\
    b_{M,1} & b_{M,2} & \cdots & b_{M,K}\\
  \end{array}
\right),
\\
 A\cdot \left(
  \begin{array}{cccc}
    g_{1,K+2} & g_{1,K+3} & \cdots & g_{1,k} \\
     g_{2,K+2} & g_{2,K+3} & \cdots & g_{2,k} \\
     \vdots & \vdots & \ddots & \vdots \\
    g_{k,K+2} & g_{k,K+3} & \cdots & g_{k,k}\\
  \end{array}
\right)=&P
\ ,
  \end{array}
\right.
\end{equation}
 has $q$ solutions by Lemma \ref{keylem}, saying $A_1,A_2,\cdots,A_q$. The solutions $A_1,A_2,\cdots,A_q$ are also solutions of System~ (\ref{equation}). Next, we show
\[
   \{\varphi_{s_{M+1}}(A_j)\,|\,j=1,2,\cdots,q\}=\f\ .
\]
Otherwise, there are two solutions $A_{j_1}$ and $A_{j_2}$ such that
\begin{equation*}
    (1,s_{M+1},s_{M+1}^2,\cdots,s_{M+1}^M)A_{j_1}\left(\begin{array}{c}
                                                                g_{1,K+1} \\
                                                                g_{2,K+1} \\
                                                                \vdots \\
                                                                g_{k,K+1}
                                                              \end{array}\right)=(1,s_{M+1},s_{M+1}^2,\cdots,s_{M+1}^M)A_{j_2}\left(\begin{array}{c}
                                                                g_{1,K+1} \\
                                                                g_{2,K+1} \\
                                                                \vdots \\
                                                                g_{k,K+1}
                                                              \end{array}\right)\ .
\end{equation*}
Then we have
\begin{equation*}
\begin{array}{rl}
    & \left(
     \begin{array}{ccccc}
      1& s_1 & s_1^2 & \cdots & s_1^{M} \\
      1& s_2 & s_2^2 & \cdots & s_2^{M}\\
      \vdots & \vdots & \vdots & \ddots & \vdots \\
     1& s_{M+1} & s_{M+1}^2 & \cdots & s_{M+1}^{M}\\
     \end{array}
   \right) A_{j_1} \left(
  \begin{array}{cccc}
    g_{1,1} & g_{1,2} & \cdots & g_{1,k} \\
     g_{2,1} & g_{2,2} & \cdots & g_{2,k} \\
     \vdots & \vdots & \ddots & \vdots \\
    g_{k,1} & g_{k,2} & \cdots & g_{k,k}\\
  \end{array}
\right)\\
= &
     \left(
     \begin{array}{ccccc}
      1& s_1 & s_1^2 & \cdots & s_1^{M} \\
      1& s_2 & s_2^2 & \cdots & s_2^{M}\\
      \vdots & \vdots & \vdots & \ddots & \vdots \\
     1& s_{M+1} & s_{M+1}^2 & \cdots & s_{M+1}^{M}\\
     \end{array}
   \right) A_{j_2} \left(
  \begin{array}{cccc}
    g_{1,1} & g_{1,2} & \cdots & g_{1,k} \\
     g_{2,1} & g_{2,2} & \cdots & g_{2,k} \\
     \vdots & \vdots & \ddots & \vdots \\
    g_{k,1} & g_{k,2} & \cdots & g_{k,k}\\
  \end{array}
\right) \ .
\end{array}
\end{equation*}
But matrices
\begin{equation*}
     \left(
     \begin{array}{ccccc}
      1& s_1 & s_1^2 & \cdots & s_1^{M} \\
      1& s_2 & s_2^2 & \cdots & s_2^{M}\\
      \vdots & \vdots & \vdots & \ddots & \vdots \\
     1& s_{M+1} & s_{M+1}^2 & \cdots & s_{M+1}^{M}\\
     \end{array}
   \right),\qquad
 \left(
  \begin{array}{cccc}
    g_{1,1} & g_{1,2} & \cdots & g_{1,k} \\
     g_{2,1} & g_{2,2} & \cdots & g_{2,k} \\
     \vdots & \vdots & \ddots & \vdots \\
    g_{k,1} & g_{k,2} & \cdots & g_{k,k}\\
  \end{array}
\right)
\end{equation*}
are invertible. So $A_{j_1}=A_{j_2}$ which contradicts to the condition $A_{j_1}\neq A_{j_2}$. And hence, the statement (1) holds.

Next, we prove (2).
Any one solution of System~(\ref{equation}) gives one $P\in \f^{(M+1)\times(k-d(C^\bot)+1)}$, while corresponding to such a $P$ there are $q$ solutions of System~(\ref{equation}) from the proof of (1). In this way, we partition solutions of System~(\ref{equation}) into $q^{k-d(C^\bot)+1}$ parts such that each part contains $q$ elements. Also from the proof of (1), the image of each part under $\varphi_{s_{M+1}}$ is $\f$. So for any $y\in \f$, the number of the inverse image of $y$ is $\#\varphi_{s_{M+1}}^{-1}(y)=q^{k-d(C^\bot)+1}$.

\end{proof}

\begin{rem}\label{motivation}
From the proofs of Lemma \ref{keylem} and Theorem \ref{general}, the coalition of malicious receivers $B$ can successfully make a substitution attack to the receiver $R_i$ if and only if $\vec{g}_i$ is contained in the subspace of $\f^k$ generated by $\{\vec{g}_j\,|\,j\in B\}$, where $\vec{g}_j$ represents the $j$-th column of the
generator matrix $G$. In this case, they can recover the private key of $R_i$. This is the motivation of the next section.
\end{rem}
Next, we give a toy example to illustrate Lemma \ref{keylem} and Theorem \ref{general}.
\begin{exm}\label{exm1}
Let $\f=\fl$. The $M=3$ messages sent are $s_1=1,s_2=2,s_3=4$. The $C$ is a systematic code with the generator matrix
\begin{equation*}
    G=\left(
        \begin{array}{ccccccccc}
          1&0&0 &0& 0& 1& 2& 4& 0\\
          0& 1& 0& 0 &0& 2 &2& 3& 2 \\
         0& 0& 1& 0 &0& 3 &1& 3& 4 \\
          0&0 &0 &1 &0 &4 &0 &0 &2 \\
          0& 0& 0& 0& 1& 2& 1& 1& 4 \\
        \end{array}
      \right)\ .
\end{equation*}
One can check that the dual code $C^\bot$ has minimum distance $d(C^\bot)=5$. The trusted authority randomly chooses $A\in \fl^{4\times5}$, for instance,
\begin{equation*}
    A=\left(
        \begin{array}{ccccc}
          3& 2& 2& 0& 2 \\
         0& 4& 3& 0& 2 \\
          0& 1& 2& 3& 1 \\
         3& 3& 0& 1& 3\\
        \end{array}
      \right)\ .
\end{equation*}
Then the trusted authority computes
\begin{equation*}
    B=AG=\left(
        \begin{array}{ccccccccc}
         3& 2& 2& 0& 2& 2& 4& 1& 0 \\
         0& 4& 3& 0& 2& 1& 3& 3& 3 \\
          0& 1& 2& 3& 1&2& 0& 0& 0 \\
         3& 3& 0& 1& 3& 4& 0& 4& 0\\
        \end{array}
      \right)
\end{equation*}
and distributes the $i$-th column of $B$ to the receiver $R_i$ as his private key.

Suppose $R_1,R_2,R_3$ are corrupted and they have seen the authenticated messages
\begin{equation*}
    \left(
       \begin{array}{c}
         x_1 \\
         x_2 \\
         x_3 \\
       \end{array}
     \right)=\left(
               \begin{array}{cccccc}
                 1 &  1& 0& 2& 4& 3 \\
                 2 & 2& 3& 1& 0& 4 \\
                 3 & 4& 4& 4& 4& 3 \\
               \end{array}
             \right)\ ,
\end{equation*}
 then they want to substitute one of the authenticated messages during the transmission by a new codeword $[s,L(s)]$ that can be accepted by one of the other receivers.
They have information about the key matrix $A$:
\begin{equation}\label{equ4}
\left\{
  \begin{array}{c}
      \left(
      \begin{array}{cccc}
        1& 1& 1& 1 \\
       1& 2& 4& 3 \\
       1& 3& 4& 2 \\
      \end{array}
    \right)A=\left(
               \begin{array}{ccccc}
                 1& 0& 2& 4& 3\\
               2& 3& 1& 0& 4 \\
                4& 4& 4& 4& 3 \\
               \end{array}
             \right), \\
     A\left(
        \begin{array}{ccc}
         1&0&0\\
         0& 1& 0 \\
         0& 0& 1 \\
         0&0 &0 \\
         0&0 &0 \\
        \end{array}
      \right)=\left(
                \begin{array}{ccc}
                  3& 2& 2 \\
                 0& 4& 3 \\
                 0& 1& 2 \\
                 3& 3& 0 \\
                \end{array}
              \right)\ .
  \end{array}
\right.
\end{equation}
This system of linear equations has $25$ solutions
\begin{equation*}
    \left(
      \begin{array}{ccccc}
        3& 0& 0& 3& 2 \\
        4& 1& 3& 2& 3 \\
        2& 0& 1& 4& 4 \\
        0& 1& 3& 0& 4 \\
      \end{array}
    \right)+a_1\left(
                 \begin{array}{ccccc}
                   0& 0& 0& 0& 0 \\
                  0& 0 &0 &0& 0 \\
                   0& 0& 1& 4 &1 \\
                  4& 0& 0& 0& 0 \\
                 \end{array}
               \right)+a_2\left(
                 \begin{array}{ccccc}
                   0& 0& 0& 0& 0 \\
                  0& 0 &0 &0& 0 \\
                   0& 0 &0 &0& 0 \\
                  0& 1& 4& 1& 4 \\
                 \end{array}
               \right)
\end{equation*}
where $a_1,a_2\in \fl$. For $s_4=4$ and any $i=4,5,\cdots,9$, we have the label map
\begin{equation*}
    \begin{array}{rccc}
      \varphi_{s_{4},R_i}: & \{\textrm{Solutions of System (\ref{equ4})}\} & \longrightarrow & \f \\
       & A & \mapsto & (1,s_{M+1},s_{M+1}^2,\cdots,s_{M+1}^M)A\left(\begin{array}{c}
                                                                g_{1,i} \\
                                                                g_{2,i} \\
                                                                \vdots \\
                                                                g_{k,i}
                                                              \end{array}\right)\ .
    \end{array}
\end{equation*}
Let $\varphi_{s_{4}}=(\varphi_{s_{4},R_4},\varphi_{s_{4},R_5},\cdots,\varphi_{s_{4},R_9})$. Then the images of $\varphi_{s_{4}}$ are
\begin{table}
\begin{equation*}
    \begin{tabular}{|c|c|c|c|c|}
       \hline
(0, 4, 3, 2, 0, 2)&
(0, 3, 1, 1, 4, 3)&
(0, 2, 4, 0, 3, 4)&
(0, 1, 2, 4, 2, 0)&
(0, 0, 0, 3, 1, 1)\\
(4, 4, 4, 2, 0, 0)&
(4, 3, 2, 1, 4, 1)&
(4, 2, 0, 0, 3, 2)&
(4, 1, 3, 4, 2, 3)&
(4, 0, 1, 3, 1, 4)\\
(3, 4, 0, 2, 0, 3)&
(3, 3, 3, 1, 4, 4)&
(3, 2, 1, 0, 3, 0)&
(3, 1, 4, 4, 2, 1)&
(3, 0, 2, 3, 1, 2)\\
(2, 4, 1, 2, 0, 1)&
(2, 3, 4, 1, 4, 2)&
(2, 2, 2, 0, 3, 3)&
(2, 1, 0, 4, 2, 4)&
(2, 0, 3, 3, 1, 0)\\
       (1, 4, 2, 2, 0, 4)&
(1, 3, 0, 1, 4, 0)&
(1, 2, 3, 0, 3, 1)&
(1, 1, 1, 4, 2, 2)&
(1, 0, 4, 3, 1, 3)\\
       \hline
     \end{tabular}
\end{equation*}
\end{table}\\
Notice that for any $i\geqslant 4$, $\varphi_{s_{4},R_i}$ is surjective and for any $y\in \f$, the number of the inverse image of $y$ is $\#\varphi_{s_4}^{-1}(y)=5$. One can check the properties of $\varphi_{s_{4},R_i}$ about surjection and uniform distribution of the images for $s_4=0$ also hold.

Actually, we can verify that even the coalition of $R_1,R_2,R_3,R_4$ can successfully generate a fraudulent codeword $[s_4,L(s_4)]$ for any other $R_{i}$ still only in a probability $1/5$ which is the success probability of randomly choosing a label from $\fl$ for a fake message.
\end{exm}
\section{Code-based Authentication Scheme and Minimal Codewords}
In the previous section, we considered that any coalition of $K$ malicious receivers can not obtain any information about any other receiver's label to make a substitution attack. To consider a weak point, we propose that for a fixed receiver $R_{i}$, what a coalition of malicious receivers that can not get any information of the label of $R_{i}$. By Theorem \ref{general}, we have seen that any coalition of up to $(d(C^\bot)-2)$ malicious receivers can not generate a valid codeword $[s,L(s)]$ for $R_{i}$ in a probability better than guessing a label from $\f$ randomly for the fake message $s$.

Denote $[V]=\{1,2,\cdots,V\}$ and $\p=\{R_{1},R_{2},\cdots,R_{V}\}$. Without any confusion, we identify the index set $\{1,2,\cdots,V\}$ and the receiver set $\{R_{1},R_{2},\cdots,R_{V}\}$.
\begin{defn}
A subset of $V-1$ receivers $\p\setminus\{R_i\}$ is call an \emph{adversary group} to $R_{i}$ if their coalition can not obtain any information of the label of $R_{i}$ when they want to make a substitution attack to $R_i$.
Define $t_i(C)$ to be the largest integer $\tau_i$ such that any subset $A\subseteq\p\setminus\{R_i\}$ with cardinality $\tau_i$ is an adversary group to $R_{i}$.
\end{defn}
\begin{defn}
A subset of $\p\setminus\{R_i\}$ that can successfully make a substitution attack to $R_{i}$ is call a \emph{substitution group} to $R_{i}$. Moreover, a substitution group is call \emph{minimal} if any one receiver is removed from the group, then the rests can not obtain any information of the label of $R_{i}$. Define $r_i(C)$ to be the smallest integer $\rho_i$ such that any subset $B\subseteq\p\setminus\{R_i\}$ with cardinality $\rho_i$ is a substitution group to $R_{i}$.
\end{defn}

 For any $A\subseteq[V]$, $\pi_A$ is the projection of $\f^V$ to $\f^{|A|}$ defined by
\[
    \pi_A((x_1,x_2,\cdots,x_V))=(x_j)_{j\in A},
\]
for any $(x_1,x_2,\cdots,x_V)\in\f^V$. And denote by $ \pi_i= \pi_{\{i\}}$ for short. For any receiver $R_{i}$, the substitution groups to $R_{i}$ are completely characterized as follows.

\begin{prop}\label{reconstruction}
For any receiver $R_{i}$, the following conditions are equivalent:
\begin{description}
\item[(i)] $B\subseteq\p\setminus\{R_i\}$ is a substitution group to $R_{i}$;

\item[(ii)] $\vec{g}_i$ is contained in the subspace of $\f^k$ generated by $\{\vec{g}_j\,|\,j\in B\}$, where $\vec{g}_j$ represents the $j$-th column of the generator matrix $G$;

\item[(iii)] there exists a codeword $\vec{c}\in C^\bot$ such that
\[
   \pi_i(\vec{c})=1\qquad \mbox{and}\qquad\pi_{B^c}(\vec{c})= \vec{0}\ ,
\]
where $B^c=(\p\setminus\{R_i\})\setminus B$ is the complement of $B$ in $\p\setminus\{R_i\}$;

\item[(iv)] there is an $\f$-linear map
    $$f_{B,i}: \pi_B(C)\longrightarrow \f$$
such that $f_{B,i}(\pi_B(\vec{c}))=\pi_i(\vec{c})$ for all $\vec{c}\in C$;

\item[(v)] there is no codeword $\vec{c}\in C$ such that
\[
   \pi_i(\vec{c})=1\qquad \mbox{and}\qquad\pi_{B}(\vec{c})= \vec{0}\ .
\]
\end{description}
\end{prop}

\begin{proof} By Remark \ref{motivation}, conditions (i) and (ii) are equivalent.

First, we show that there exists a codeword $\vec{c}\in C^\bot$ such that $\pi_i(\vec{c})\neq 0$. If not, that is, for any codeword $\vec{c}\in C^\bot$, it holds $\pi_i(\vec{c})= 0$.
Then the unit vector with the unique nonzero component $1$ on the $i$-th coordinate belongs to $C$, which contradicts to the assumption $d(C)\geqslant2$.

So there exists a codeword $\vec{c}\in C^\bot$ such that $\pi_i(\vec{c})=1$ by the linearity of $C$. The rest of the proof that conditions (ii) and (iii) are equivalent is clear.

(iii)$\Longrightarrow$(iv). For any codeword $\vec{y}\in C^\bot$ with
\[
   \pi_i(\vec{y})=1\qquad \mbox{and}\qquad\pi_{B^c}(\vec{y})= \vec{0}\ ,
\]
we have
\[
    \sum_{j\in B}\pi_j(\vec{y})\pi_j(\vec{c})+\pi_i(\vec{c})=0
\]
for any codeword $\vec{c}\in C$. So define $f_{B,i}: \pi_B(C)\rightarrow \f$ by setting
\[
    f_{B,i}(\pi_B(\vec{c}))=- \sum_{j\in B}\pi_j(\vec{y})\pi_j(\vec{c}),
\]
for all $\vec{c}\in C$. Then $f_{B,i}$ satisfies the condition.

(iv)$\Longrightarrow$(iii). From the proof of ``(iii)$\Longrightarrow$(iv)", we see that the required codeword in $C^\bot$ is actually the coefficients of
the map
\[
   \phi_{B,i}=\pi_{i}-f_{B,i}.
\]

(iv)$\Longrightarrow$(v). If the statement (v) does not hold, then there exists a codeword $\vec{c}\in C$ such that
\[
   \pi_i(\vec{c})=1\qquad \mbox{and}\qquad\pi_{B}(\vec{c})= \vec{0}\ ,
\]
which contradicts to $f_{B,i}(\pi_B(\vec{c}))=\pi_i(\vec{c})$.

(v)$\Longrightarrow$(iv). A map
    $$f_{B,i}: \pi_B(C)\longrightarrow \f$$
satisfying $f_{B,i}(\pi_B(\vec{c}))=\pi_i(\vec{c})$ for all $\vec{c}\in C$ is always linear over $\f$ by the linearity of $C$. So if the map
    $$f_{B,i}: \pi_B(C)\longrightarrow \f$$
satisfying $f_{B,i}(\pi_B(\vec{c}))=\pi_i(\vec{c})$ for all $\vec{c}\in C$ does not exist, then there exist two different codewords $\vec{c},\vec{c}'\in C$ such that
\[
   \pi_i(\vec{c})\neq\pi_i(\vec{c}') \qquad \mbox{and}\qquad\pi_{B}(\vec{c})= \pi_{B}(\vec{c}')\ .
\]
That is, the codeword $\vec{x}=\vec{c}-\vec{c}'\in C$ satisfies
\[
   \pi_i(\vec{x})=\pi_i(\vec{c}-\vec{c}')\neq 0 \qquad \mbox{and}\qquad\pi_{B}(\vec{x})= \pi_{B}(\vec{c}-\vec{c}')=\vec{0}\ ,
\]
which contradicts to (v).
\end{proof}

By Proposition \ref{reconstruction}, adversary groups to $R_{i}$ can be completely characterized by
\begin{prop}\label{adv}
For any receiver $R_{i}$, the following conditions are equivalent:
\begin{description}
\item[(i)] $A\subseteq\p\setminus\{R_i\}$ is an adversary group to $R_{i}$;

\item[(ii)] $\vec{g}_i$ is not contained in the subspace of $\f^k$ generated by $\{\vec{g}_j\,|\,j\in A\}$;

\item[(iii)] there is no codeword $\vec{c}\in C^\bot$ such that
\[
   \pi_i(\vec{c})=1\qquad \mbox{and}\qquad\pi_{A^c}(\vec{c})=\vec{0}\ ;
\]

\item[(iv)] there exists a codeword $\vec{c}\in C$ such that
\[
   \pi_i(\vec{c})=1\qquad \mbox{and}\qquad\pi_{A}(\vec{c})= \vec{0}\ .
\]
\end{description}
\end{prop}

\begin{cor}\label{cor1}
\begin{description}
\item[(i)] For any $i=1,2,\cdots,V$, we have
\[
   d(C^\bot)-1\leqslant r_i(C)\leqslant V-d(C)+1\ ,
\]
and
\[
   \max\{r_i(C)\,|\,i=1,2,\cdots,V\}=V-d(C)+1, \quad\min\{r_i(C)\,|\,i=1,2,\cdots,V\}-1=d(C^\bot)-1\ .
\]

\item[(ii)] For any $i=1,2,\cdots,V$, we have
\[
    d(C^\bot)-2\leqslant t_i(C)\leqslant r_i(C)-1\ ,
\]
and
\[
   \min\{t_i(C)\,|\,i=1,2,\cdots,V\}=d(C^\bot)-2\ .
\]
\end{description}
\end{cor}

\begin{proof} (i) Suppose $B\subseteq\p\setminus\{R_i\}$ is any substitution group to $R_i$. By Proposition \ref{reconstruction} (iii), there is a codeword $\vec{c}\in C^\bot$ such that
\[
   \pi_i(\vec{c})=1\qquad \mbox{and}\qquad\pi_{B^c}(\vec{c})=\vec{0}\ .
\]
Then we have
\[
   d(C^\bot)\leqslant\mathrm{wt}(\vec{c})\leqslant |B|+1\ .
\]
So
\[
 r_i(C)\geqslant |B|\geqslant d(C^\bot)-1\ .
\]

For any $B\subseteq\p\setminus\{R_i\}$ with cardinality $\geqslant V-d(C)+1$, it is obvious that any codeword $\vec{c}\in C$ with $\pi_{i}(\vec{c})=1$ (in the proof of Proposition \ref{reconstruction}, we have seen that such a codeword does exist.) has $\pi_{B}(\vec{c})\neq \vec{0}$. Otherwise, the minimum distance $d(C)\leqslant V-(V-d(C)+1)=d(C)-1$. So by Proposition \ref{reconstruction} (v), it follows
\[
    r_i(C)\leqslant V-d(C)+1\ .
\]

Let $\vec{c}$ be a codeword in $C$ with minimum Hamming weight. Denote by $S$ the support of $\vec{c}$. Let $B=[V]\setminus S$. Then by Proposition \ref{reconstruction} (v), $B$ is not a substitution group to $R_i$ for any $i\in S$. So
\[
   \max\{r_i(C)\,|\,i=1,2,\cdots,V\}\geqslant\max\{r_i(C)\,|\,i\in S\}\geqslant |B|+1= V-d(C)+1\ .
\]
And hence
\[
   \max\{r_i(C)\,|\,i=1,2,\cdots,V\}=V-d(C)+1\ .
\]
To prove $\min\{r_i(C)\,|\,i=1,2,\cdots,V\}-1=d(C^\bot)-1$, it suffices to show
\[
   r_i(C)=d(C^\bot)-1
\]
for some $i=1,2,\cdots,V$. Let $\vec{y}$ be a codeword in $C^\bot$ with minimum Hamming weight. Denote by $T$ the support of $\vec{y}$. For any $i\in T$, $T\setminus\{i\}$ is a substitution group to $R_i$ with cardinality $d(C^\bot)-1$. On the other hand, by Proposition \ref{reconstruction} (ii), any subset of $\p\setminus\{R_i\}$ with cardinality $\leqslant d(C^\bot)-2$ could not be a substitution group to $R_i$. So
\[
   r_i(C)=d(C^\bot)-1
\]
for any $i\in T$.

(ii) $t_i(C)\leqslant r_i(C)-1$ by the definition. For any $B\subseteq\p\setminus\{R_i\}$ with cardinality $\leqslant d(C^\bot)-2$, there is no codeword $\vec{c}\in C^\bot$ such that
\[
   \pi_i(\vec{c})=1\qquad \mbox{and}\qquad\pi_{B^c}(\vec{c})=\vec{0}\ .
\]
If not, then there is a codeword $\vec{c}\in C^\bot$ such that
\[
   \pi_i(\vec{c})=1\qquad \mbox{and}\qquad\pi_{B^c}(\vec{c})=\vec{0}\ .
\]
Then $C^\bot$ has a codeword $\vec{c}$ with Hamming weight $\leqslant |B|+1(\leqslant d(C^\bot)-1)$ which is impossible. So by Proposition \ref{adv}, $B$ is an adversary group to $R_i$. And hence
\[
    d(C^\bot)-2\leqslant t_i(C)\ .
\]
Since
\[
    d(C^\bot)-2\leqslant t_i(C)\leqslant r_i(C)-1\ ,
\]
we have
\[
    d(C^\bot)-2\leqslant \min\{t_i(C)\,|\,i=1,2,\cdots,V\}\leqslant \min\{r_i(C)\,|\,i=1,2,\cdots,V\}-1=d(C^\bot)-2\ .
\]
So
\[
   \min\{t_i(C)\,|\,i=1,2,\cdots,V\}=d(C^\bot)-2\ .
\]
\end{proof}

By Corollary \ref{cor1}, it is natural to get
\begin{cor}\label{cor2}For any receiver $R_{i}$, we have
\begin{description}
\item[(i)] Subsets of $\p\setminus\{R_i\}$ with cardinality $\geqslant (V-d(C)+1)$ are substitution groups to $R_{i}$.

\item[(ii)] Subsets of $\p\setminus\{R_i\}$ with cardinality $\leqslant (d(C^\bot)-2)$ are adversary groups to $R_{i}$.

\item[(iii)] For MDS codes $C$, subsets of $\p\setminus\{R_i\}$ with cardinality $\leqslant (d(C^\bot)-2)$ are all the adversary groups to $R_{i}$.
\end{description}
\end{cor}

There is a gap in Corollary~\ref{cor2} in general we do not known whether a subset of size in the gap is a substitution group to $R_{i}$ or not for general code-based authentication scheme. Actually, it is {\bf NP}-hard to list all substitution groups to $R_{i}$ in general. Even for authentication scheme based on algebraic geometric codes from elliptic curves, it is already {\bf NP}-hard (under {\bf RP}-reduction) to list all substitution groups to $R_{i}$ \cite{CLX,chengqi}.

 By Proposition \ref{reconstruction}, we obtain the main result of this section, a generalization of Proposition~\ref{minimal}:
\begin{thm}\label{thm}For the authentication scheme we constructed, we have
\begin{description}
\item[(i)] The set of all minimal substitution groups to the receiver $R_{i}$ is determined completely by all the minimal codewords respect to $i$ in $C^\bot$.

\item[(ii)] All adversary groups to the receiver $R_{i}$ are one-to-one corresponding to subsets of $[V]\setminus\{i\}$ such that each of them together with $i$ does not contain any support of minimal codeword respect to $i$ in $C^\bot$.
\end{description}
\end{thm}

\begin{exm}
Continue with Example~\ref{exm1}, minimum codewords respect to $5$ in $C^\bot$ are list in Table~1.
\begin{table}
\begin{equation*}
    \begin{tabular}{|c|c|c|c|c|}
      \hline
    (2 2 1 0 1 0 4 0 0)&
    (0 1 0 3 1 2 0 2 2)&
    (0 3 0 0 1 3 3 4 4)&
    (3 0 0 4 1 1 3 0 1)&
    (2 0 0 0 1 2 4 2 1)\\
    (3 0 2 0 1 0 2 2 0)&
    (0 0 3 1 1 0 2 4 2)&
    (0 4 0 1 1 1 2 0 0)&
    (0 0 4 4 1 2 4 0 4)&
    (4 0 2 4 1 4 1 0 0)\\
    (1 0 1 4 1 0 2 0 3)&
    (0 0 1 0 1 1 3 2 3)&
    (0 2 2 0 1 4 3 0 2)&
    (1 0 0 1 1 3 0 4 1)&
    (2 0 3 4 1 3 0 0 2)\\
    (0 2 0 4 1 0 4 3 3)&
    (4 0 0 3 1 0 2 3 1)&
    (0 0 2 3 1 3 0 3 0)&
    (1 4 0 0 1 0 1 3 0)&
    (3 3 0 2 1 0 0 3 4)\\
    (4 3 3 0 1 0 0 4 0)&
    (3 1 4 2 1 2 0 0 0)&
    (0 3 1 3 1 0 0 0 1)&
    (0 0 0 2 1 4 1 1 1)&
    (2 4 0 4 1 4 0 1 0)\\
    (4 2 0 0 1 1 0 0 3)&
    (4 0 4 0 1 3 0 2 4)&
    (0 1 4 0 1 0 3 1 0)&&\\
      \hline
    \end{tabular}
\end{equation*}
\caption{}
\end{table}

Since coordinate $5$ is in the support of any minimal codeword respect to $5$, we exclude $5$ from all supports of these codewords. Then we get supports of minimal codewords excluding $5$, see Table~2.
\begin{table}
\begin{equation*}
    \begin{tabular}{|c|c|c|c|c|c|c|}
      \hline
    \{ 1, 2, 3, 7 \}&
    \{ 2, 4, 6, 8, 9 \}&
    \{ 2, 6, 7, 8, 9 \}&
    \{ 1, 4, 6, 7, 9 \}&
    \{ 1, 6, 7, 8, 9 \}\\
    \{ 1, 3, 7, 8 \}&
    \{ 3, 4, 7, 8, 9 \}&
    \{ 2, 4, 6, 7 \}&
    \{ 3, 4, 6, 7, 9 \}&
    \{ 1, 3, 4, 6, 7 \}\\
    \{ 1, 3, 4, 7, 9 \}&
    \{ 3, 6, 7, 8, 9 \}&
    \{ 2, 3, 6, 7, 9 \}&
    \{ 1, 4, 6, 8, 9 \}&
    \{ 1, 3, 4, 6, 9 \}\\
    \{ 2, 4, 7, 8, 9 \}&
    \{ 1, 4, 7, 8, 9 \}&
    \{ 3, 4, 6, 8 \}&
    \{ 1, 2, 7, 8 \}&
    \{ 1, 2, 4, 8, 9 \}\\
    \{ 1, 2, 3, 8 \}&
    \{ 1, 2, 3, 4, 6 \}&
    \{ 2, 3, 4, 9 \}&
    \{ 4, 6, 7, 8, 9 \}&
    \{ 1, 2, 4, 6, 8 \}\\
    \{ 1, 2, 6, 9 \}&
    \{ 1, 3, 6, 8, 9 \}&
    \{ 2, 3, 7, 8 \}&&\\
      \hline
    \end{tabular}
\end{equation*}
\caption{}
\end{table}

So any substitution group to the receiver $R_{5}$ must contain at least one set in Table~2.
And subset of $\p\setminus\{R_5\}$ that does not contain any one set in Table~1 can not make a substitution attack to the receiver $R_{5}$ successfully in a probability better than $1/5$ ($1/5$ is the success probability of randomly choosing a label from $\fl$ for a fake message) using their knowledge of the key matrix $A$. From Table~2, notice that most subsets of $\p\setminus\{R_5\}$ with cardinality $4$ can not generate a fake message that can accepted by $R_{5}$ successfully in a probability better than $1/5$, even subsets of $\p\setminus\{R_5\}$ with cardinality $5$ can not, such as shown in Table~3. While in the scheme of Safavi-Naini and Wang with a $[9,5]$ Reed-Solomon code (the field must have cardinality $\geqslant 9$), any subset of $\p\setminus\{R_5\}$ with cardinality $5$ can successfully recover the private key of $R_5$ and hence they can easily make a substitution attack to $R_5$.

\begin{table}
\begin{equation*}
    \begin{tabular}{|c|c|c|c|c|c|}
      \hline
    \{ 1, 2, 3, 8, 9 \}&
    \{ 2, 3, 4, 6, 7 \}&
    \{ 3, 4, 6, 8, 9 \}&
    \{ 1, 3, 4, 6, 7 \}&
    \{ 1, 2, 3, 4, 7 \}&
    \{ 3, 4, 7, 8, 9 \}\\
    \{ 2, 3, 6, 8, 9 \}&
    \{ 1, 3, 6, 8, 9 \}&
    \{ 2, 4, 6, 7, 9 \}&
    \{ 1, 4, 6, 7, 8 \}&
    \{ 2, 4, 6, 7, 8 \}&
    \{ 1, 4, 6, 7, 9 \}\\
    \{ 1, 2, 4, 7, 8 \}&
    \{ 1, 2, 4, 7, 9 \}&
    \{ 1, 2, 6, 7, 8 \}&
    \{ 1, 2, 6, 7, 9 \}&
    \{ 1, 2, 3, 4, 9 \}&
    \{ 1, 3, 6, 7, 9 \}\\
    \{ 2, 3, 6, 7, 8 \}&
    \{ 1, 3, 6, 7, 8 \}&
    \{ 2, 3, 6, 7, 9 \}&
    \{ 1, 2, 4, 6, 8 \}&
    \{ 1, 2, 4, 6, 7 \}&
    \{ 1, 2, 3, 4, 8 \}\\
    \{ 2, 3, 4, 6, 9 \}&
    \{ 4, 6, 7, 8, 9 \}&
    \{ 1, 3, 4, 6, 8 \}&
    \{ 1, 2, 4, 6, 9 \}&
    \{ 2, 4, 7, 8, 9 \}&
    \{ 1, 2, 3, 7, 8 \}\\
    \{ 1, 2, 3, 7, 9 \}&
    \{ 2, 6, 7, 8, 9 \}&
    \{ 1, 2, 6, 8, 9 \}&
    \{ 1, 4, 7, 8, 9 \}&
    \{ 2, 3, 4, 6, 8 \}&
    \{ 3, 6, 7, 8, 9 \}\\
    \{ 3, 4, 6, 7, 9 \}&
    \{ 1, 2, 3, 6, 7 \}&
    \{ 1, 3, 4, 6, 9 \}&
    \{ 3, 4, 6, 7, 8 \}&
    \{ 1, 2, 4, 8, 9 \}&
    \{ 1, 3, 4, 8, 9 \}\\
    \{ 1, 6, 7, 8, 9 \}&
    \{ 1, 2, 3, 6, 9 \}&
    \{ 1, 2, 7, 8, 9 \}&
    \{ 2, 3, 4, 8, 9 \}&
    \{ 2, 3, 4, 7, 8 \}&
    \{ 1, 3, 4, 7, 9 \}\\
    \{ 1, 4, 6, 8, 9 \}&
    \{ 2, 4, 6, 8, 9 \}&
    \{ 2, 3, 4, 7, 9 \}&
    \{ 1, 3, 4, 7, 8 \}&
    \{ 1, 2, 3, 6, 8 \}&
    \{ 1, 3, 7, 8, 9 \}\\
    \{ 2, 3, 7, 8, 9 \}&
    \{ 1, 2, 3, 4, 6 \}&&&&\\
      \hline
    \end{tabular}
\end{equation*}
\caption{}
\end{table}

\end{exm}
\section{Conclusion}
In this paper, we construct an authentication scheme for multi-receivers and multiple messages based on linear code $C\,[V,k,d]$. There are many advantages. Compared with schemes based on MACs or digital signatures which depend on computational security. Our scheme is an unconditional secure authentication scheme, which can offer robustness against a coalition of up to ($d(C^\bot)-2$) malicious receivers. Similarly as the generalization of Shamir's secret sharing scheme to linear secret sharing sceme based on linear codes, compared with the scheme of Safavi-Naini and Wang~\cite{SW} which has a constraint on the number of verifying receivers that can not be larger than the size of the finite field. Our scheme allows arbitrary receivers for a fixed message base field. And, for some important receiver, coalitions of $k$ or more malicious receivers can not yet make a substitution attack on the receiver more efficiently than randomly guessing a label from the finite field for a fake message. While the authentication scheme of Safavi-Naini and Wang is a $(V,k)$ threshold authentication scheme, any $k$ of the $V$ receivers can easily produce a fake message that can be accepted by the receiver.

\bibliographystyle{plain}
\bibliography{authentication}

\end{document}